\documentclass{article}
\usepackage[main, final]{style}

\usepackage[utf8]{inputenc} 
\usepackage[T1]{fontenc}    
\usepackage{hyperref}       
\usepackage{url}            
\usepackage{booktabs}       
\usepackage{amsfonts}       
\usepackage{nicefrac}       
\usepackage{microtype}      
\usepackage{xcolor}         
\usepackage{outlines}
\usepackage{graphicx}

\usepackage{amsmath}
\usepackage{amsthm}

\numberwithin{equation}{section}

\theoremstyle{plain}
\newtheorem{theorem}{Theorem}[section]
\newtheorem{lemma}[theorem]{Lemma}

\theoremstyle{definition}
\newtheorem{definition}[theorem]{Definition}

\theoremstyle{remark}

\newcommand{\E}{\mathbb{E}}

\title{Sparse Random Matrices for Dimensionality Reduction}

\author{
  Pierre Mackenzie\\
  Department of Computer Science\\
  University of British Columbia\\
  \texttt{lardet[dot]pierre[at]gmail.com} \\
}

\begin{document}

\maketitle

\begin{abstract}
The Johnson-Lindenstrauss (JL) theorem states that a set of points in high-dimensional space can be embedded into a lower-dimensional space while approximately preserving pairwise distances with high probability~\cite{OgJL}. The standard JL theorem uses dense random matrices with Gaussian entries. However, for some applications, sparse random matrices are preferred as they allow for faster matrix-vector multiplication. I outline the constructions and proofs introduced by \citet{binarycoins} and the contemporary standard by \citet{sparser}. Further, I implement and empirically compare these sparse constructions with standard Gaussian JL matrices.
\end{abstract}
\section{Introduction}
\subsection{Motivation}

Much of the data in the modern world lives in high dimensions. Whether it be images, text embeddings, counts, or countless other modalities, data is often represented as vectors in a high-dimensional space. However, high-dimensional data can be challenging to work with, not least because of computational and storage costs. Dimensionality reduction techniques aim to address these challenges by mapping high-dimensional data to a lower-dimensional space while preserving important properties.

\cite{OgJL} introduced the JL lemma, which provides a theoretical foundation for randomised dimensionality reduction. They proved that one can construct a random linear map $R\in \mathbb{R}^{k\times d}$ that approximately preserves pairwise distances between points with high probability. However, multiplication by a dense random matrix can be computationally expensive. This motivated the so-called `Fast-JL' transforms~\cite{FastJL}, which leveraged carefully structured random matrices to achieve multiplication in $O(d \log d)$ rather than $O(dk)$ time.

However, many real-world datasets are not only high-dimensional but also sparse, meaning that most of their entries are zero. Such sparsity can be found in text data represented as bag-of-words vectors, user-item interaction matrices in recommendation systems, or in streaming application where received inputs have only one non-zero entry. It would be advantageous to have JL transforms that can exploit sparsity in the input vectors to reduce computational costs. This led to the development of sparse JL transforms~\cite{binarycoins,sparse,sparser}, which use random matrices with many zero entries while preserving the same probabilistic guarantee.

In this project, I explore the important works in sparse JL transforms. I focus on the seminal work by~\citet{binarycoins} and the modern standard by~\citet{sparser}. I present their constructions and proofs in Section~\ref{sec:jl_transforms}, and I empirically evaluate their performance in Section~\ref{sec:experiments}. 

\subsection{Preliminaries}\label{sec:preliminaries}

\begin{theorem}[Markov Inequality]
    Let \(X\) be a non-negative random variable. Then, for any \(a > 0\):
    \begin{equation}
        \Pr[X \geq a] \leq \frac{\E[X]}{a}
    \end{equation}
\end{theorem}

\begin{definition}[Frobenius Norm]
    For a matrix \(A \in \mathbb{R}^{n \times n}\), the Frobenius norm is defined as:
    \begin{equation}
        \|A\|_F = \sqrt{\sum_{i,j} |A_{ij}|^2}
    \end{equation}
\end{definition}

\begin{definition}[Operator Norm]
    For a matrix \(A \in \mathbb{R}^{n \times n}\), the operator norm is defined as:
    \begin{equation}
        \|A\|_2 = \sup_{||x||_2 = 1} ||Ax||_2
    \end{equation}
\end{definition}

\begin{theorem}[Hanson-Wright Inequality]\label{thm:hanson-wright}
    Let $z = (z_1, z_2, \ldots, z_n) \in \mathbb{R}^n$ be a vector of i.i.d. random variables which are $\pm 1$ with equal probability. Let $B$ be a symmetric $n \times n$ matrix. Then, for every $t \geq 2$, we can bound the expectation of the quadratic form $z^T B z$ as follows:
    \begin{equation}
        \E[|z^T B z - \E[z^T B z]|^t] \leq C^t \cdot \max\left\{ \sqrt{t} \cdot \|B\|_F, t \cdot \|B\|_2 \right\}^t
    \end{equation}
    for some constant $C > 0$.
\end{theorem}

{
    \small\textbf{Reference:}~\cite{hasonwright}
}

\section{JL Transforms}\label{sec:jl_transforms}

In this section, I explore three different constructions of Johnson-Lindenstrauss (JL) transforms for dimensionality reduction. Each construction offers a different level of sparsity, but offers the same guarantee on preservation of pairwise distances. Below, I give an explanation of the proof techniques used to establish each construction's JL properties.

\subsection{The Standard JL Transform}\label{sec:standard_jl}

\begin{theorem}[Johnson-Lindenstrauss Transform]\label{thm:jl}
    Let vectors \(x_1, x_2, \ldots, x_n \in \mathbb{R}^d \) be given. Choose \( \epsilon \in (0, 1) \) and let $k=O(\log n / \epsilon^2)$. Construct a random matrix \( R \in \mathbb{R}^{k \times d} \) where each entry \( R_{ij} \) is drawn independently from \( \mathcal{N}(0, 1/k) \). Let $y_i = Rx_i \forall i$. Then, with probability at least \( 1 - \frac{1}{n} \), for all \(i, j \in [n]\):
\begin{equation}\label{eq:jl_guarantees}
\begin{gathered}
(1-\epsilon)\|x_i\|_2^2 \le \|y_i\|_2^2 \le (1+\epsilon)\|x_i\|_2^2\\
(1-\epsilon)\|x_i-x_j\|_2^2 \le \|y_i - y_j\|_2^2 \le (1+\epsilon)\|x_i-x_j\|_2^2
\end{gathered}
\end{equation}
\end{theorem}

{
\small
\textbf{Reference}:~\cite{OgJL}
}

The proofs of the theorem and lemma below are adapted from~\cite{secondrandomized}.

In essence, JL states that a random projection preserves norms and pairwise distances within a factor of \( (1 \pm \epsilon) \) with high probability. $\epsilon$ controls the desired accuracy of distance preservation and determines the target dimension \(k\). This theorem is often also stated in terms of the existence of a deterministic linear map \( f: \mathbb{R}^d \to \mathbb{R}^k \) that satisfies the same distance-preserving properties. 

The essence of the JL transform comes from the following lemma, which provides a probabilistic guarantee for the preservation of the norm of any unit vector under random projection.

\begin{lemma}[Distributional JL]\label{lemma:distributional_jl}
    Choose $\epsilon, \delta \in (0,1)$. There exists $k= O(\epsilon^{-2}\log(1/\delta))$ such that a random matrix $R \in \mathbb{R}^{k \times d}$ where each entry \( R_{i,j} \) is drawn independently from \( \mathcal{N}(0, 1/k) \) such that for any vector $v \in \mathbb{R}^d$ where $\|v\|_2 = 1$, we have
    \begin{equation}
    \Pr\left[||Rv||_2^2 \in (1 - \epsilon, 1 + \epsilon)\right] \geq 1 - \delta
    \end{equation}
\end{lemma}

\begin{proof}[Sketch of Proof of Theorem~\ref{thm:jl}]
    From Lemma~\ref{lemma:distributional_jl}, we can quickly derive Theorem~\ref{thm:jl}. 
    
    First, note that if 
    $||Rv||_2^2 \in (1 - \epsilon, 1 + \epsilon)$ for a unit vector $v\in\mathbb{R}^d$, then for any vector $u = cv$ for some $c\in \mathbb{R}$, we have
    \[
    ||Ru||_2^2 \in (1 - \epsilon)||u||_2^2, (1 + \epsilon)||u||_2^2
    \]
    That is, if the unit sphere's norm is preserved within a factor of \( (1 \pm \epsilon) \), then the norm of any vector $u$ is preserved by the same multiplicative factor.

    Secondly, there are $O(n^2)$ pairs of points among the $n$ input vectors. By applying a union bound over all $n$ vectors and pairs of vectors, we can ensure that with probability at least \( 1 - \frac{1}{n} \) that the distance between every pair of points is preserved within $(1\pm \epsilon)$. Letting $\delta = 1/n^3$ in Lemma~\ref{lemma:distributional_jl} and applying the union bound over each individual vector and each pair of vectors yields the result. 
\end{proof}

In fact, any distance preservation with a polynomial number of combinations of vectors can be preserved with high probability using the same technique. If there are $O(n^c)$ combinations of vectors, we can set $\delta = 1/n^{c+1}$ in Lemma~\ref{lemma:distributional_jl} and apply a union bound to ensure that all combinations are preserved with probability at least $1 - 1/n$. Because $\delta$ only appears inside a logarithm in the required target dimension $k$, this only increases the required target dimension by a constant factor; $k$ remains $O(\log n / \epsilon^2)$.


\begin{proof}[Sketch of Proof of Lemma~\ref{lemma:distributional_jl}]

    Each entry of $Rv$ is a dot product between a row of $R$ and the vector $v$. Since each entry of $R$ is drawn from $\mathcal{N}(0, 1/k)$, each entry of $Rv$ is drawn from $\mathcal{N}(0, ||v||_2^2/k) = \mathcal{N}(0, 1/k)$ (since we assumed $||v||_2 = 1$). Thus, each entry of $Rv$ is an independent Gaussian random variable with mean 0 and variance $1/k$.

    Therefore, the squared norm $||Rv||_2^2$ is the sum of $k$ independent random variables, each distributed as $\mathcal{N}(0, 1/k)^2$. Each such variable has expectation $1/k$, so the expectation of $||Rv||_2^2$ is 1. We now wish to show that $||Rv||_2^2$ concentrates around its expectation.

    We note that the sum of $k$ independent $\mathcal{N}(0, 1)^2$ random variables follows a chi-squared distribution with $k$ degrees of freedom: $\chi_k$. Further noting that scaling a random variable by a constant $c$ scales its variance by $c^2$, we see that $||Rv||_2^2\sim \frac{1}{k}\chi_k$. Hence $k ||Rv||_2^2\sim \chi_k$.

    We use a bound on the probability that $X\sim\chi_k$ deviates from its mean. For any $\epsilon \in (0,1)$,
    \begin{equation}
        \Pr[X \notin (1 - \epsilon, 1 + \epsilon)k] \leq 2\exp(-k \epsilon^2 / 8)
    \end{equation}
    
    {\small
    \textbf{Reference:} \cite{UnderstandingMLShai}, B.7
    }

    We can substitute $X = k ||Rv||_2^2$ to get
    \begin{equation}
    \Pr[||Rv||_2^2 \notin (1 - \epsilon, 1 + \epsilon)] \leq 2\exp(-k \epsilon^2 / 8)
    \end{equation}
    Setting the right-hand side to be at most $\delta$ and solving for $k$ yields Lemma~\ref{lemma:distributional_jl}.
\end{proof}

We see that JL works by averaging multiple independent estimates of the norm of a vector, leading to concentration. With enough samples (i.e., a large enough target dimension $k$), we can ensure that the norm is preserved with high probability. The required target dimension $k$ only grows logarithmically with $n$, making JL useful even when $n$ is large.

We also note that $k$ is agnostic to the original dimension $d$. Thus, JL can be useful for both large $n$ and large $d$.

However, the standard JL transform has a drawback: the random matrix \( R \) is dense, meaning that every entry is non-zero almost surely. A dense matrix is inevitable when each entry is drawn from a Gaussian distribution. The proof above relies on independent Gaussians in order to characterise the distribution of $||Rv||_2^2$ as a chi-squared distribution, from which we derive concentration bounds. To be able to use other distributions for the entries of \( R \), we need different proof techniques.

\subsection{A Discrete, Sparse Matrix}\label{sec:achlioptas}

{\small
\textbf{Reference}:~\cite{binarycoins}
}

To achieve a sparse matrix $R$, we must have positive probability of zero entries. This means we cannot use continuous distributions like the Gaussian.~\citet{binarycoins} proposes two different distributions for each entry of $R$.

\begin{minipage}{0.45\linewidth}
\begin{equation}\label{eq:binarycoins}
R_{ij} =
\begin{cases}
1  & \text{with probability } 1/2 \\
-1 & \text{with probability } 1/2
\end{cases}
\end{equation}
\end{minipage}
\hfill
\begin{minipage}{0.49\linewidth}
\begin{equation}\label{eq:sparsebinarycoins}
R_{ij} = \sqrt{3}\times
\begin{cases}
+1 & \text{with probability } 1/6 \\
0  & \text{with probability } 2/3 \\
-1 & \text{with probability } 1/6
\end{cases}
\end{equation}
\end{minipage}

Equation~\ref{eq:binarycoins} is a discrete distribution akin to flipping coins. This is what gives the paper its name: `Johnson-Lindenstrauss with binary coins'. Equation~\ref{eq:sparsebinarycoins} gives a sparse distribution, where two-thirds of the entries are zero in expectation. Both distributions have mean 0 and variance 1 and yield JL transforms with identical guarantees. The proof of both follows a nearly identical structure.

\begin{theorem}[Sparse JL]\label{thm:jl_binary_coins}
    Given a set of vectors \(x_1, x_2, \ldots, x_n \in \mathbb{R}^d \) and $\epsilon \in (0,1)$, there exists $k=O(\log n / \epsilon^2)$ such that for a random matrix $R \in \mathbb{R}^{k \times d}$ with every entry of $R$ drawn independently from distribution~\ref{eq:binarycoins} or~\ref{eq:sparsebinarycoins}, and setting $y_i = \frac{1}{\sqrt{k}} R x_i$, then with probability at least \( 1 - \frac{1}{n} \), for all \(i, j \in [n]\), the JL guarantee~\ref{eq:jl_guarantees} is satisfied.
\end{theorem}

The proof uses the standard technique of quantifying the probability of deviation from the mean via exponentiation and Markov's inequality. We first require a lemma which bounds expressions related to each entry of the output vector. 

\begin{lemma}\label{lemma:binary_coins_moment_bounds}

Let $R_j$ denote the $j$-th row of $R$ with entries independently drawn from either distribution~\ref{eq:binarycoins} or~\ref{eq:sparsebinarycoins}. For all $h \in [0, 1/2)$ and any unit vector $v$:
\begin{equation}\label{eq:binarycoins_moment_bounds}
    \E[\exp(h(R_j v)^2)] \leq \frac{1}{\sqrt{1 - 2h}}
\end{equation}
\begin{equation}\label{eq:binarycoins_fourth_moment}
    \E[(R_j v)^4] \leq 3
\end{equation}
\end{lemma}

Using this lemma, we can prove Theorem~\ref{thm:jl_binary_coins}.

\begin{proof}[Proof of Theorem~\ref{thm:jl_binary_coins}]
    As before with Lemma~\ref{lemma:distributional_jl},it suffices to prove that the guarantee holds for any vector $v \in \mathbb{R}^d$ with $\|v\|_2 = 1$, with probability at least $1 - \delta$ for some $\delta \in (0,1)$, where the target dimension scales logarithmically with $1/\delta$. The full theorem then follows by applying a union bound over all vectors and pairs of vectors, as in the proof of Theorem~\ref{thm:jl}. 
    
    Thus, we focus on providing an analogue of Lemma~\ref{lemma:distributional_jl} for this proof. To simplify notation, let the sum $S = \sum_{j=1}^k (R_j v)^2$. Note that $\E[S] = k$ since $\E[(R_j v)^2] = 1$ for all $j$.
    An analogue to Lemma~\ref{lemma:distributional_jl} must show that for any $\epsilon \in (0,1)$, 
    \begin{equation}
        \Pr[S \notin (1 - \epsilon)k, (1 + \epsilon)k] \leq \delta
    \end{equation}
     for some $\delta$ that decreases exponentially with $k \epsilon^2$.
    

    First, we bound $\Pr[S > (1 + \epsilon)k]$. By Markov's inequality on the random variable $e^{h S}$ for any $h > 0$, we have
    \begin{equation}
        \Pr[S > (1 + \epsilon)k] = \Pr[e^{h S} > e^{h (1 + \epsilon)k}] \leq \E[e^{h S}] e^{-h (1 + \epsilon)k}
    \end{equation}

    Since the rows of $R$ are i.i.d., we have
    \begin{equation}
        \E[e^{h S}] = \prod_{j=1}^k \E[\exp(h (R_j v)^2)] = \left( \E[\exp( h (R_1 v)^2)] \right)^k
    \end{equation}

    Using the inequality from ~\ref{eq:binarycoins_moment_bounds}, we have
    \begin{equation}
        \Pr[S > (1 + \epsilon)k] \leq \left(\frac{1}{\sqrt{1 - 2h}} \right)^k e^{-h (1 + \epsilon)k}
    \end{equation}

    Setting $h = \frac{1}{2} \frac{\epsilon}{1 + \epsilon} < \frac{1}{2}$ minimizes the right-hand side, yielding
    \begin{equation}
        \Pr[S > (1 + \epsilon)k] \leq ((1 + \epsilon) e^{-\epsilon})^{k/2} = \exp\left(\frac{k}{2}(\log(1 + \epsilon) - \epsilon )\right)
    \end{equation}

    Finally, using the inequality $\log(1 + x) \leq x - x^2/2 + x^3/3$ from the Taylor expansion of $\log(1+x)$ for $|x| < 1$ and using that for $\epsilon\in (0,1)$ we have $-\epsilon^2 (1/2 - \epsilon/3) \leq -\epsilon^2/6$, we get
    \begin{equation}
        \Pr[S > (1 + \epsilon)k] \leq \exp(-k \epsilon^2 / 12)
    \end{equation}

    Up to constants inside the exponential, this is the same bound as in the distributional JL lemma~\ref{lemma:distributional_jl}.

    Next, we bound $\Pr[S < (1 - \epsilon)k]$. Similar to before, by Markov's inequality, for any $h > 0$:

    \begin{equation}
        \Pr[S < (1 - \epsilon)k] \leq (\E[\exp(-h (R_1 v)^2)])^k \exp({h (1 - \epsilon)k})
    \end{equation}

    Using the 2nd order Taylor expansion of $\exp(-h (R_1 v)^2)$ around $h=0$, the fact that $\E[(R_1 v)^2] = 1$, and the bound from~\ref{eq:binarycoins_fourth_moment}, we have
    \begin{equation}
    \begin{gathered}
            \E[\exp(-h (R_1 v)^2)] \le \E\left[1 - h (R_1 v)^2 + \frac{h^2 (R_1 v)^4}{2}\right]\\
            \le 1 - h + \frac{3 h^2}{2} \le \exp(-h + \frac{3 h^2}{2})
    \end{gathered}  
    \end{equation}

    Substituting this back into the bound on $\Pr[S < (1 - \epsilon)k]$, we get
    \begin{equation}
    \begin{gathered}
        \Pr[S < (1 - \epsilon)k] \leq \left(\exp\left(-h + \frac{3 h^2}{2}\right)\right)^k \exp(h (1 - \epsilon)k)\\
        = \exp\left(k h \left( \frac{3 h}{2} - \epsilon \right) \right)
    \end{gathered} 
    \end{equation}

    Setting $h = \epsilon / 3$ minimizes the right-hand side, giving
    \begin{equation}
        \Pr[S < (1 - \epsilon)k] \leq \exp(-k \epsilon^2 / 6)
    \end{equation}

    Combining both bounds, we have
    \begin{equation}
        \Pr[S \notin (1 - \epsilon)k, (1 + \epsilon)k] \leq \exp(-k \epsilon^2 / 12) + \exp(-k \epsilon^2 / 6) \leq 2 \exp(-k \epsilon^2 / 12)
    \end{equation}

    Thus, $S$ concentrates around its mean with the same exponential dependence on $k \epsilon^2$ as in the distributional JL lemma~\ref{lemma:distributional_jl}. 
\end{proof}

This proof shows that we can achieve JL guarantees using discrete and sparse distributions for the entries of the random matrix \( R \). The key insight is that the moment bounds from Lemma~\ref{lemma:binary_coins_moment_bounds} allow us to control the tail behavior of the sum of squared projections, enabling the application of Markov's inequality with tight tails. I now summarise the key ideas behind the proof of the moment bounds in Lemma~\ref{lemma:binary_coins_moment_bounds}.

\begin{proof}[Intuition for Proof of Lemma~\ref{lemma:binary_coins_moment_bounds}]
Both inequalities rely on comparing to the corresponding expressions using a Gaussian random variable. In particular, we require that all even moments of the discrete distributions are at most those of a Gaussian random variable with the same variance. That is, let $G \sim \mathcal{N}(0,1)$ be a standard normal random variable. We require that for all $m \geq 1$ and any unit vector $v$:

\begin{equation}\label{eq:even_moment_comparison}
    \E[(R_j v)^{2m}] \leq \E[G^{2m}] \text{ for all } m \geq 1
\end{equation}

\citet{binarycoins} show this by demonstrating that the `worst case' vector for the discrete distributions~\ref{eq:binarycoins} and~\ref{eq:sparsebinarycoins} is a vector with all entries equal in magnitude. This can be intuitively understood as spreading out the randomness as much as possible to maximise higher moments. The variance is fixed, but the higher moments can be increased by spreading out the mass.

They then show that~\ref{eq:even_moment_comparison} boils down to showing that the even moments of distributions~\ref{eq:binarycoins} and~\ref{eq:sparsebinarycoins} on the worst-case vectors are at most those of a Gaussian random variable. This is trivially satisfied by distribution~\ref{eq:binarycoins}, since all odd moments are zero and all even moments are 1, and can be shown for distribution~\ref{eq:sparsebinarycoins} by induction. For more details, see~\cite{binarycoins}.

From here, we can prove inequality~\ref{eq:binarycoins_fourth_moment} using~\ref{eq:even_moment_comparison} with $m=2$ to find $\E[(R_j v)^4] \leq \E[G^4] = 3$.

For inequality~\ref{eq:binarycoins_moment_bounds}, we upper bound $\E[\exp(h (R_j v)^2)]$ by its Taylor expansion, and using the even moment comparison from~\ref{eq:even_moment_comparison} to bound each term of the expansion by the corresponding term for a Gaussian random variable. That is,

\begin{equation}
    \E[\exp(h (R_j v)^2)] = \sum_{m=0}^{\infty} \frac{h^m}{m!} \E[(R_j v)^{2m}] \leq \sum_{m=0}^{\infty} \frac{h^m}{m!} \E[G^{2m}] = \E[\exp(h G^2)]
\end{equation}

By evaluating the integral we find $\E[\exp(h G^2)] = \frac{1}{\sqrt{1 - 2h}}$, yielding inequality~\ref{eq:binarycoins_moment_bounds}. 
    
\end{proof}

This shows that any distribution with even moments bounded by those of a Gaussian random variable can be used to construct a JL transform using the same proof technique. Therefore, in a sense, the Gaussian distribution itself is the `worst-case' distribution for JL under this proof technique.~\citet{binarycoins} use a more detailed derivation to show that the constants are indeed better for the discrete distributions than for the Gaussian distribution. The asymptotics, however, remain the same.

\subsection{A Sparser Matrix}\label{sec:sparser}

{\small
\textbf{Reference}: \cite{sparser}
}

The discrete distribution~\ref{eq:sparsebinarycoins} yields a sparse JL-transform with two-thirds of the entries zero in expectation. Intuitively, it should be possible to have the number of non-zero entries be dependent on the distortion accuracy $\epsilon$. The larger the allowed error, the fewer non-zero entries should be required to achieve the JL guarantee.

\citet{sparser} provide two similar constructions of JL transforms which satisfy this intuition. The proof techniques for both constructions are nearly identical. Further, they provide two different analyses which yield different asymptotic bounds on the number of non-zero entries. Here, I only provide details for the `graph construction' with the simpler analysis.

\begin{definition}[Graph construction]\label{def:graph_construction}
Let $R \in \mathbb{R}^{k \times d}$. For each column $i \in [d]$, choose a subset of rows
$S_i \subset [k]$ uniformly at random with $|S_i| = s$.
Set
\[
R_{r,i} =
\begin{cases}
\pm \tfrac{1}{\sqrt{s}} & \text{if } r \in S_i,\\
0 & \text{otherwise},
\end{cases}
\]
where the signs are independent and $\pm 1$ with probability 1/2. Intuitively, each column of $R$ has exactly $s$ non-zero entries. Each column's non-zero entries are chosen by selecting $s$ rows uniformly at random without replacement. Each non-zero entry is $\pm 1/\sqrt{s}$ with equal probability.
\end{definition}

As before, to prove a JL-style theorem, we need only prove that the norm is not distorted too much for any single vector $v$ with $\|v\|_2 = 1$ and apply the union bound to get the full JL guarantee for a set of $n$ vectors.

\begin{theorem}[Sparser JL Theorem]\label{thm:sparser_jl}
Let $R$ be constructed using the graph construction~\ref{def:graph_construction} with $k=O(\epsilon^{-2}\log(1/\delta))$ and $s = O(\epsilon^{-1} \log(1/\delta)\log(d/\delta))$. Then, for any unit vector $v \in \mathbb{R}^d$,
\begin{equation}
    \Pr[||Rv||_2^2 \notin (1 - \epsilon, 1 + \epsilon)] \leq \delta
\end{equation}
\end{theorem}

\begin{proof}[Proof]
    The theorem statement is equivalent to
    \begin{equation}
        \Pr\left[ \left| ||Rv||_2^2 - 1 \right| > \epsilon \right] \leq \delta
    \end{equation}
We introduce some notation for simplicity. Let $R_{r,i} = \sigma_{r,i} \eta_{r,i} / \sqrt{s}$, where $\sigma_{r,i} \in \{ -1, +1\}$ is a random sign and $\eta_{r,i} \in \{0,1\}$ indicates whether entry $(r,i)$ is non-zero. Then,
\begin{equation}
    ||Rv||_2^2 - 1 = \sum_{r=1}^k \left( \sum_{i=1}^d R_{r,i} v_i \right)^2 = \frac{1}{s} \sum_{r=1}^k \left( \sum_{i=1}^d \sigma_{r,i} \eta_{r,i} v_i \right)^2 - 1
\end{equation}
\begin{equation}
    = \frac{1}{s} \sum_{r=1}^k \left( \sum_{i=1}^d \sigma_{r,i}^2 \eta_{r,i}^2 v_i^2 + \sum_{i \neq j \in [d]} \sigma_{r,i} \sigma_{r,j} \eta_{r,i} \eta_{r,j} v_i v_j \right) - 1
\end{equation}

Note that $\sigma_{i,j}^2 = 1$ almost surely and $\eta_{i,j}^2 = \eta_{i,j}$. Further, since each column has exactly $s$ non-zero entries, we have $\sum_{i=1}^k \eta_{i,j} = s$ for all $j$. Therefore we define $Z$, the deviation from the mean, as
\begin{equation}
    Z:=||Rv||_2^2 - 1 = \frac{1}{s} \sum_{r=1}^k \sum_{i \neq j \in [d]} \sigma_{r,i} \sigma_{r,j} \eta_{r,i} \eta_{r,j} v_i v_j
\end{equation}

We can think of $Z$ as summing the contributions of deviance from the mean from each row of $R$. Such deviations come from `collisions', where two non-zero entries in the same row interact. If there were as many rows as columns, there could be no collisions, but this would require $k = d$. With $k < d$, if we want to sample all $d$ input coordinates, we must use $s > 1$ and sample each coordinate multiple times, leading to collisions.

In expectation, the deviance is zero, since the random signs $\sigma_{i,j}$ are independent and have mean zero. However, we need to bound the probability that the deviance is large in magnitude. To bound its tail probability, we apply a Markov bound on the $m$-th moment of $|Z| \ge 0$ for some even integer $m$. Note that $|Z|^m = Z^m$ since $m$ is even. Thus, for any $\epsilon > 0$,

\begin{equation}
    \Pr[|Z| > \epsilon] = \Pr[Z^m > \epsilon^m] \leq \epsilon^{-m}\E[Z^m]
\end{equation}

In order to apply the Hanson-Wright theorem, we now express $Z$ as a quadratic form in the random signs $\sigma_{i,j}$ such that $Z = \sigma^T A \sigma$, where $\sigma\in\mathbb{R}^{kd}$ is the vector of all random signs and $A$ is a matrix dependent on the indicators $\eta_{i,j}$ and the vector $v$. Specifically, $A$ is a block-diagonal matrix with $k$ blocks of size $d \times d$, where the $i$-th block $A_i$ has entries $(A_i)_{j,l} = \eta_{i,j} \eta_{i,l} v_j v_l / s$ for $j \neq l$ and zeros on the diagonal. We assume the $\eta_{i,j}$ are fixed for now and later return to constraints that they must satisfy.

We now apply the Hanson-Wright theorem~\ref{thm:hanson-wright} to bound the $m$-th moment of $Z$. Note that $\E[\sigma^T A \sigma] = \E[Z] = 0$ since the random signs have mean zero. Thus, for any integer $m \geq 1$,
\begin{equation}
    \Pr[|Z| > \epsilon] \leq C^m \epsilon^{-m} \max\{\sqrt{m} ||A||_F, m ||A||_2\}^m
\end{equation}

We now require bounds on the Frobenius norm $||A||_F$ and the operator norm $||A||_2$ of the matrix $A$. These are provided by the following lemmas.

\begin{lemma}[Collision Bound]\label{lemma:collision_bound}
    With $s=O(\epsilon^{-1} \log(1/\delta) \log(d/\delta))$, then with probability at least $1-\delta/2$, the set of indicators $\eta_{i,j}$ satisfies that for all $i \neq j \in [d]$, $\sum_{r=1}^k \eta_{r,i} \eta_{r,j} = O(s^2/k)$.
\end{lemma}

\begin{lemma}[Norm Bounds]\label{lemma:norm_bounds}
If for all $i \neq j \in [d]$ we have $\sum_{r=1}^k \eta_{r,i} \eta_{r,j} = O(s^2/k)$, then
\begin{equation}
\|A\|_F^2 = O(1/k)
\quad\text{and}\quad
\|A\|_2 \le \frac{1}{s}.
\end{equation}

\end{lemma}
The proof of Lemma~\ref{lemma:collision_bound} can be found in Proof~\ref{proof:collision_bound}. For a short proof of the norm bounds Lemma~\ref{lemma:norm_bounds}, see~\cite{sparser}. Now applying these bounds to the previous inequality, we have
\begin{equation}
    \Pr[|Z| > \epsilon] \leq (C/\epsilon \max\{O(\sqrt{m/k}), m/s\})^m
\end{equation}

Choosing $m \ge \log(2/\delta)$ and setting $s = O(\epsilon^{-1} \log(1/\delta))$ and $k = O(\epsilon^{-2} \log(1/\delta))$ then we find $\sqrt{m/k} = O(\epsilon)$ and $m/s = O(\epsilon)$. Thus, for some set of constants, we can make $C/\epsilon \max\{\sqrt{m/k}, m/s\} \leq 1/2$, yielding
\begin{equation}
    \Pr[|Z| > \epsilon] \leq (1/2)^m \leq (1/2)^{\log(2/\delta)} = \delta/2
\end{equation}

Applying a union bound over this failure probability and the failure probability of Lemma~\ref{lemma:collision_bound} gives a total failure probability of at most $\delta$, completing the proof.
\end{proof}

\citet{sparser} omit the details of Lemma~\ref{lemma:collision_bound} as they later prove of a tighter version of Theorem~\ref{thm:sparser_jl} with $s = O(\epsilon^{-1} \log(1/\delta))$. They derive this bound from first principles over the randomness of both the $\eta_{i,j}$ and $\sigma_{i,j}$. This analysis is much longer and more involved so is omitted here. Instead, I provide a proof for Lemma~\ref{lemma:collision_bound}.

\begin{proof}[Proof of Collision Bound Lemma~\ref{lemma:collision_bound}]\label{proof:collision_bound}

    Fix columns $i,j$ where $i\neq j$. Let $X_r$ denote the indicator random variable for a collision in row $r$. That is, $X_r = \eta_{r,i} \eta_{r,j}$. Define $X:=\sum_{r=1}^k X_r = \sum_{r=1}^k \eta_{r,i} \eta_{r,j}$ to be the total number of collisions. We wish to show that $X = O(s^2/k)$ with high probability.
    
    Notice that $\E[X_i] = \Pr[\eta_{r,i} = 1 \text{ and } \eta_{r,j} = 1] = s^2 / k^2$, since each column has $s$ non-zero entries chosen uniformly at random. Therefore, $\E[X] = \sum_{r=1}^k \E[X_r] = k \cdot s^2 / k^2 = s^2 / k$. In fact, $X\sim\text{Hypergeometric}(k, s, s)$, since we are sampling $s$ rows without replacement for each column from a total of $k$ rows with $s$ `successes' (non-zero entries) in each column.
    
    This motivates the use of a Chernoff bound to show concentration around the mean. Though the $X_i$ are not independent, they are negatively dependent, and \citet{hypergeomtail} show that Chernoff-style bounds hold for draws without replacement. Thus we use a one-sided Chernoff bound to show concentration below the mean,
    \begin{equation}
        \Pr[X > 2 s^2 / k] \leq \exp(-c s^2 / k)
    \end{equation}
    for some constant $c > 0$. We now apply a union bound over all pairs of columns. There are at most $d^2$ such pairs, so
    \begin{equation}
        \Pr[\exists i \neq j \in [d], X > 2 s^2 / k] \leq d^2 \exp(-c s^2 / k)
    \end{equation}

    Setting the right-hand side to be less than $\delta/2$ and solving for $s$ yields $s = O(\epsilon^{-1}\log(d/\delta)\log(1/\delta))$.
\end{proof}
\section{Experiments}\label{sec:experiments}

To compare and investigate the performance of the various Johnson-Lindenstrauss transforms described in Section~\ref{sec:jl_transforms}, I conduct a series of experiments focusing on two key aspects: sparsity and norm preservation. These experiments aim to evaluate how well each JL transform maintains the geometric properties of the original data while also considering the computational efficiency afforded by sparsity.

\subsection{Setup}

For each experiment, I generate $n=5,000$ input vectors in $\mathbb{R}^d$ where $d=10,000$. The vectors are generated in two different ways. A first set are sampled uniformly at random from the unit sphere in $\mathbb{R}^d$. A second set are sparse, where each vector has only $t$ non-zero entries, with the positions of the non-zero entries chosen uniformly at random. All vectors are of unit magnitude though, as we have seen, the distortion of a linear map is invariant under scaling of the input vectors. I select $k=50$, $s=16$ and $t=5$ for all experiments unless otherwise stated. 


I then sample each of the three JL transforms described in Section~\ref{sec:jl_transforms} and apply them to the set of input vectors. In all experiments, I sample $30$ independent instances of each JL transform and, where relevant, report the mean and standard deviation over instances.

I calculate the \emph{distortion} of each JL transform on the each input vectors where distortion is defined as \[
    \Delta := \|Rx\|_2^2 - 1,
\] since all input vectors are of unit magnitude. I then report various statistics on the distortion across all input vectors and JL transform instances.

I refer to the three JL transforms used as \emph{Dense}, \emph{Ach} and \emph{Sparse} where \emph{Dense} refers to the standard JL with entries drawn from independent Gaussians introduced in Section~\ref{sec:standard_jl}, \emph{Ach} refers to the transform from~\citet{binarycoins} discussed in Section~\ref{sec:achlioptas} and \emph{Sparse} refers to the sparsest transform from~\citet{sparser} discussed in Section~\ref{sec:sparser}. 

\subsection{Results}

\textbf{1. Matrix sparsity $s$}: I first investigate the effect of varying $s$ on the distribution of distortion for the \emph{Sparse} construction. A plot showing quantiles of distortion while varying $s$ for both sparse and dense input vectors is shown in Figure~\ref{fig:sparsity_vs_distortion}. The plot also shows the distortion quantiles for the \emph{Ach} construction as a reference. Note that the \emph{Dense} construction has very similar distortion to the \emph{Ach} construction and is therefore omitted for clarity.

In the plot, we see that the \emph{Sparse} transform does better as we increase $s$ but only with sparse vectors with the exception of the median, where $s=1$ dominates. For dense vectors, all values of $s$ and the \emph{Ach} transform perform nearly identically.

\begin{figure}[ht]
  \centering
  \includegraphics[width=\linewidth]{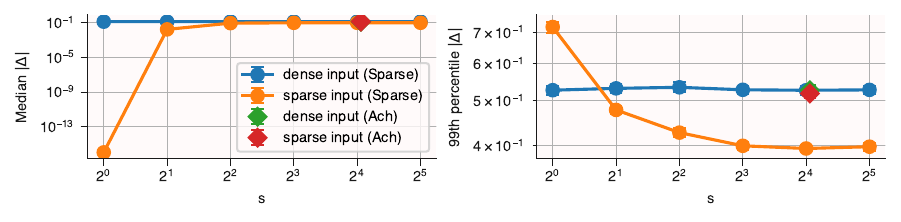}
  \caption{Various Quantiles of distortion vs sparsity \( s \) for the Sparse JL transform on both sparse and dense input vectors. The means and standard deviations across 30 independent instances are plotted. The Achlioptas transform is shown as a reference. 
  }\label{fig:sparsity_vs_distortion}
\end{figure}

\textbf{2. Input sparsity $t$}: To further investigate why the only notable difference in distribution is seen with sparse input vectors, I vary the sparsity of the input vectors while keeping \( s=16 \) fixed for the \emph{Sparse} transform. The results are shown in Figure~\ref{fig:input_sparsity_vs_distortion}. We see that as the input vectors become denser, the distortion quantiles converge to those of the \emph{Ach} transform. This suggests that the \emph{Sparse} transform is particularly well-suited for preserving norms of sparse input vectors, while its advantage diminishes with denser inputs.

\begin{figure}[ht]
  \centering
  \includegraphics[width=\linewidth]{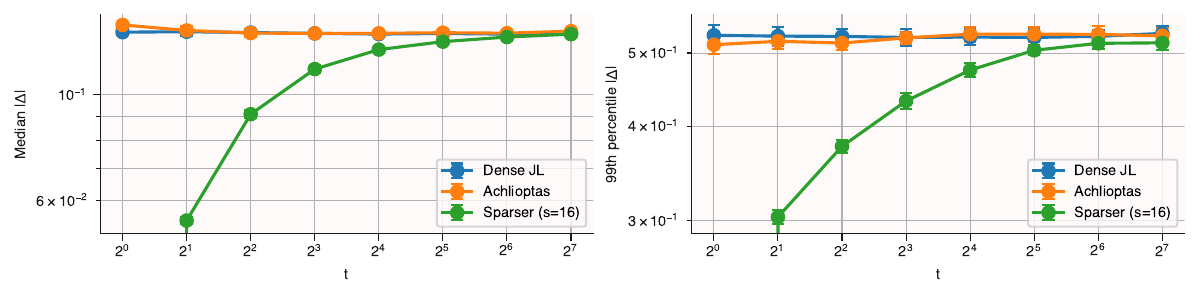}
  \caption{The median and 99th percentile of distortion vs input vector sparsity for the \emph{Sparse} transform with \( s=16 \). 
  }\label{fig:input_sparsity_vs_distortion}
\end{figure}

The quantile is not plotted for $t=1$. This is because distortion is always 0 since there are a fixed $s$ non-zero entires per column. When $t=1$, only one column from $R$ is used, and there are no collisions, thus the input is reconstructed perfectly. Increasing to $t=2$ introduces the possibility of collisions, leading to non-zero distortion. As we increase $t$, the likelihood of collisions increases, resulting in higher distortion quantiles. The same logic does not apply to the \emph{Ach} transform since we have no guarantee on the number of non-zero entries per column. This is what leads to the difference in performance between the two transforms for very sparse input vectors, and the convergence as the input vectors become denser.

\textbf{3. Distortion Distribution}: In order to compare the distortion across all transforms, the cumulative distribution function (CDF) of distortion is plotted in Figure~\ref{fig:cdf_s} for only sparse input vectors. The distributions are nearly identical for sparse vectors. Here, we can clearly see the effect of varying $s$. With lower $s$ there is a greater chance of few or no collisions, thus the median absolute distortion is lower. However, when there are distortions, they tend to be larger, leading to a heavier tail. 

\begin{figure}[ht]
  \centering
  \includegraphics[width=\linewidth]{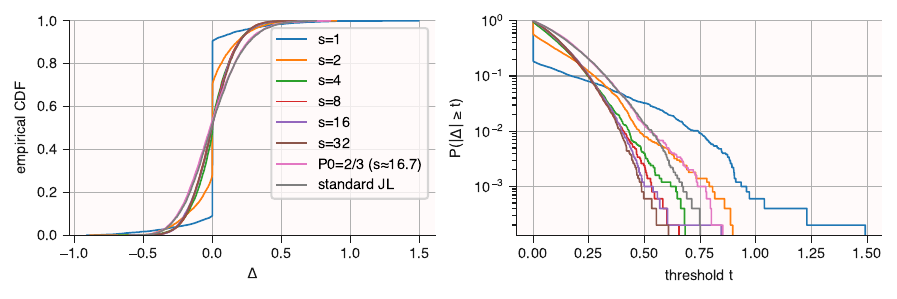}
  \caption{CDF of distortion for all three JL transforms on dense input vectors. The left plot shows the full CDF while the right plot shows the tail on a logarithmic scale.
  }\label{fig:cdf_s}
\end{figure}

\textbf{4. Target dimension $k$}: As a final experiment, I vary $k$ while keeping $s=16$ fixed for the \emph{Sparse} transform and plot the median and 99th percentile of distortion in Figure~\ref{fig:distortion_vs_k_sparse_vs_dense}. The results are shown for both dense and sparse input vectors. All results are very similar where the transform's distortion decreases as \( k \) increases. Note the log-scale on the y-axis which leads to a linear relationship between distortion and target dimension \( k \) as we would expect from the exponential tail bounds. There is one exception: \emph{Sparse} with sparse input vectors has a lower distortion across the board, for the same reasons discussed previously.

\begin{figure}[ht]
  \centering
  \includegraphics[width=\linewidth]{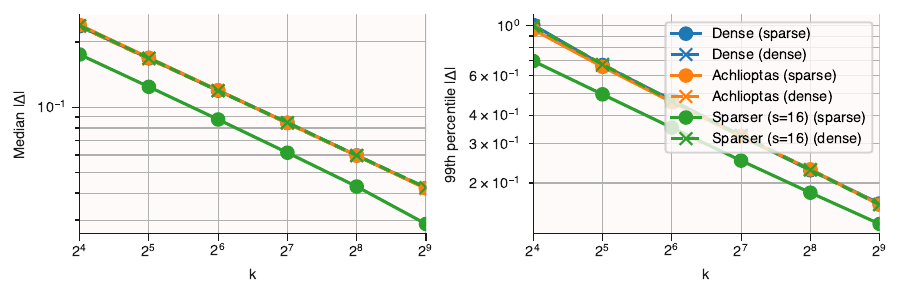}
  \caption{Median and 99th percentile of distortion vs target dimension \( k \) for both sparse and dense input vectors. Note the log-scale leading to a linear relationship.
  }\label{fig:distortion_vs_k_sparse_vs_dense}
\end{figure}
\section{Conclusion}\label{sec:conclusion}

In this project, I have explored various Johnson-Lindenstrauss transforms with a focus on their proofs, constructions, and empirical performance. Efforts over the last 40 years have led to increasingly efficient JL transforms, particularly those which can handle sparse inputs effectively. Experiments here have demonstrated the claimed advantages of sparse JL transforms when dealing with sparse input data, showcasing their ability to preserve norms. A different set of experiments could be conducted to empirically evaluate computational efficiency, or to contrast sparse JL transforms with other dimensionality reduction techniques in a variety of applications.

\newpage
\bibliographystyle{plainnat}
\bibliography{references}

@misc{sparser,
  title         = {Sparser Johnson-Lindenstrauss Transforms},
  author        = {Daniel M. Kane and Jelani Nelson},
  year          = {2014},
  eprint        = {1012.1577},
  archiveprefix = {arXiv},
  primaryclass  = {cs.DS},
  url           = {https://arxiv.org/abs/1012.1577}
}

@misc{sparse,
  title         = {A Sparse Johnson--Lindenstrauss Transform},
  author        = {Anirban Dasgupta and Ravi Kumar and Tamás Sarlós},
  year          = {2010},
  eprint        = {1004.4240},
  archiveprefix = {arXiv},
  primaryclass  = {cs.DS},
  url           = {https://arxiv.org/abs/1004.4240}
}

@article{binarycoins,
  title    = {Database-friendly random projections: Johnson-Lindenstrauss with binary coins},
  journal  = {Journal of Computer and System Sciences},
  volume   = {66},
  number   = {4},
  pages    = {671-687},
  year     = {2003},
  note     = {Special Issue on PODS 2001},
  issn     = {0022-0000},
  doi      = {https://doi.org/10.1016/S0022-0000(03)00025-4},
  url      = {https://www.sciencedirect.com/science/article/pii/S0022000003000254},
  author   = {Dimitris Achlioptas}
}

@book{UnderstandingMLShai,
  author    = {Shai Shalev-Shwartz and Shai Ben-David},
  title     = {Understanding Machine Learning: From Theory to Algorithms},
  publisher = {Cambridge University Press},
  year      = {2014},
  isbn      = {9781107057135},
  doi       = {10.1017/CBO9781107298019}
}

@incollection{OgJL,
  author    = {William B. Johnson and Joram Lindenstrauss},
  title     = {Extensions of Lipschitz mappings into a Hilbert space},
  booktitle = {Contemporary Mathematics},
  volume    = {26},
  pages     = {189--206},
  year      = {1984},
  publisher = {American Mathematical Society},
  doi       = {10.1090/conm/026/737400}
}

@book{secondrandomized,
  author      = {Nick Harvey},
  title       = {A second course in randomized algorithms},
  year        = {2025},
  institution = {University of British Columbia},
  note        = {Available online at https://www.cs.ubc.ca/~nickhar/Book2.pdf}
}

@article{hypergeomtail,
  author  = {Chv{\'a}tal, Va{\v{s}}ek},
  title   = {The Tail of the Hypergeometric Distribution},
  journal = {Discrete Mathematics},
  volume  = {25},
  number  = {3},
  pages   = {285--287},
  year    = {1979}
}

@inproceedings{FastJL,
  author    = {Ailon, Nir and Chazelle, Bernard},
  title     = {Approximate Nearest Neighbors and the Fast Johnson--Lindenstrauss Transform},
  booktitle = {Proceedings of the 38th Annual ACM Symposium on Theory of Computing (STOC)},
  year      = {2006},
  pages     = {557--563}
}

@misc{hasonwright,
  title         = {Hanson-Wright inequality and sub-gaussian concentration},
  author        = {Mark Rudelson and Roman Vershynin},
  year          = {2013},
  eprint        = {1306.2872},
  archiveprefix = {arXiv},
  primaryclass  = {math.PR},
  url           = {https://arxiv.org/abs/1306.2872}
}



\end{document}